\documentclass[10pt,conference]{IEEEtran}
\usepackage{cite}
\usepackage{amsmath,amssymb,amsfonts}
\usepackage{algorithmic}
\usepackage{graphicx}
\usepackage{textcomp}
\usepackage{xcolor}
\usepackage{xfrac}
\usepackage{amsmath}
\usepackage{amsfonts}
\usepackage{amssymb}%
\usepackage{todonotes}
\providecommand{\U}[1]{\protect\rule{.1in}{.1in}}
\def\BibTeX{{\rm B\kern-.05em{\sc i\kern-.025em b}\kern-.08em
T\kern-.1667em\lower.7ex\hbox{E}\kern-.125emX}}

\usepackage[utf8]{inputenc}
\usepackage{cite}
\usepackage{hyperref}
\usepackage{amsmath,amssymb,amsfonts,amsthm}
\usepackage{algorithmic}
\usepackage{graphicx}
\usepackage{textcomp}
\usepackage{cancel}
\usepackage{xcolor}
\def\BibTeX{{\rm B\kern-.05em{\sc i\kern-.025em b}\kern-.08em
    T\kern-.1667em\lower.7ex\hbox{E}\kern-.125emX}}
\usepackage{xfrac}
\usepackage{enumerate}
\providecommand{\U}[1]{\protect\rule{.1in}{.1in}}

\newtheorem*{ass}{Attack model}

\newcommand{\C}{\mathbb{C}}

\newcommand{\R}{\mathbb{R}}

\newcommand{\calD}{\mathcal{D}}

\newcommand{\calL}{\mathcal{L}}

\newcommand{\calP}{\mathcal{P}}

\newcommand{\abs}[1]{\left\vert #1 \right\vert}
\newcommand{\norm}[1]{\Vert #1 \Vert}

\newcommand{\sprod}[1]{\left\langle #1 \right\rangle}

\usepackage{tikz}
\usetikzlibrary{cd}

\newcommand{\impl}{\Rightarrow}

\usepackage{dsfont}

\newcommand{\geqsim}{\gtrsim}

\usepackage{dsfont}

\newtheorem{lem}{Lemma}

\newtheorem{theo}[lem]{Theorem}

\newtheorem{defi}{Definition}

\theoremstyle{remark}
\newtheorem{rem}{Remark}

\definecolor{gerhard}{rgb}{1,0,0}

\begin{document}
\title{Bilinear Compressive Security}

\author{
\IEEEauthorblockN{Axel Flinth} \\
\IEEEauthorblockA{\textit{ Mathematics and Mathematical Statisics} \\ \textit{Umeå University} \\
Umeå, Sweden \\
axel.flinth@umu.se\\}
\and
\IEEEauthorblockN{Hubert Orlicki} \\
\IEEEauthorblockA{\textit{Dep. of Computer Science} \\\textit{Cracow University of Technology} \\
Cracow, Poland\\
hubert.orlicki@petra-po.pl \\}
\and
\IEEEauthorblockN{Semira Einsele, Gerhard Wunder} \\
\IEEEauthorblockA{ \textit{Dep. of Computer Science} \\ \textit{Freie Universität Berlin} \\
Berlin, Germany \\
semira.einsele@fu-berlin.de \\} }

\maketitle

\begin{abstract} 

Beyond its widespread application in signal and image processing, \emph{compressed sensing} principles have been greatly applied to secure information transmission (often termed 'compressive security'). In this scenario, the measurement matrix $Q$ acts as a one time pad encryption key (in complex number domain) which can achieve perfect information-theoretic security together with other benefits such as reduced complexity and energy efficiency particularly useful in IoT. However, unless the matrix is changed for every message it is vulnerable towards known plain text attacks: only $n$ observations suffices  to recover a key $Q$ with $n$ columns. In this paper, we invent and analyze a new method (termed 'Bilinear Compressive Security (BCS)') addressing these shortcomings: In addition to the linear encoding of the message $x$ with a matrix $Q$, the sender convolves the resulting vector with a randomly generated filter $h$. Assuming that $h$ and $x$ are sparse, the receiver can then recover $x$ without knowledge of $h$ from $y=h*Qx$ through blind deconvolution. We study a rather idealized known plaintext attack for recovering $Q$ from repeated observations of $y$'s for different, known $x_k$, with varying and unknown $h$ ,giving Eve a number of advantages not present in practice. Our main result fro BCS states that under a weak symmetry condition on the filter $h$, recovering $Q$ will require extensive sampling from transmissions of $\Omega\left(\max\left(n,(n/s)^2\right)\right)$ messages $x_k$ if are $s$-sparse. Remarkably, with $s=1$ it is impossible to recover the key. In this way, the scheme is much safer than standard compressed sensing even though our assumptions are much in favor towards a potential attacker.

\end{abstract}

\section{Introduction}



\emph{Compressed sensing} \cite{FouRau2013} refers to the idea that underdetermined linear systems $y=Qx$ can be solved efficiently under the assumption that the vector $x$ is sparse. In compressed sensing, $Q$ is called measurement matrix and vector $y$ is a measurement. A typical result is as follows: If $Q \in \C^{m,n}$ is drawn at random according to a standard Gaussian distribution with $m\geqsim s\log(n)$, any $s$-sparse $x_0\in \C^n$ can be uniquely recovered from $Qx_0$ with high probability, for instance via $\ell_1$-minimization.

The above very naturally suggests a communication protocol: If Alice wants to send a sparse vector $x\in \C^n$ to Bob, she compresses with $Q$, and only sends the $m$-dimensional measurements $y$ to Bob. As long as Bob also knows the matrix $Q$, he can subsequently apply a compressed sensing recovery algorithm to decompress $x$ without loss.

How secure is this communication protocol? That is, if an eavesdropper Eve gains knowledge about the measurement $y$ (cyphertext), but is oblivious to the secret measurement matrix $Q$ (secret key) except that it is drawn from a normal Gaussian distribution, what can be inferred about the message $x$ (plaintext)? The short answer is only its norm, i.e., essentially nothing. Formally, we have $\mathbb{P}({y \,\vert \, x }) = \mathbb{P}(y, \, \vert \, \norm{x} )$\cite{bianchi2014security}. Also in cases where leakage of information about the norm is critical, one can argue that compressive sensing has a form of \emph{computational security} -- in essence, there will be far too many sparse solutions of $y=Qx$ that are close to the groundtruth
\cite{rachlin2008secrecy}. There have been many other works in this direction -- we refer to \cite{testa2019compressed} for an overview -- where it is sometimes termed \emph{compressive security}. Interestingly, very recently, compressive security has found application in IoT due to favorable properties such as reduced computational complexity and energy efficiency \cite{HealthIoT_CS}.


\subsection{Known plaintext attacks}
While appealing, the above scheme rigorously only works in a \emph{one shot setting}. If the measurement matrix $Q$ is re-used over several communication rounds, the communication is vulnerable to a so-called \emph{known/chosen plaintext attack}: If an eavesdropper Eve can observe pairs $(x_k, Qx_k)$ where the $x_k$'s are carefully selected, she will of course gain information about the key 
$Q$. If only one or a few such pairs are revealed, the scheme is still essentially uncompromised -- this is analyzed in detail in \cite{Cambaberi2015} --  but if Eve can observe many pairs, the scheme is easily cracked: In fact, if $n$ linearly independent $y_k$'s are observed, Even can compute the secret key by solving a system of linear equations. 
Hence, if perfect security is aspired, Alice and Bob need to necessarily renew their keys from time to time. 

There have been attempts to devise schemes that effectively accomplish the change of measurement matrix in subsequent rounds. Alice and Bob can for instance change the complex vector space basis in which the signal $x$ is sparse in every round of communications \cite{Zhang2016Bilevel}. This are includes so-called scrambling approaches. If these bases, that need to be agreed upon by Alice and Bob are revealed, however, the scheme again becomes vulnerable. Hence, a better approach is needed.

\subsection{Bilinear compressive security}
In this paper, we will consider a novel way of concealing the key from the attacker 
in different communication rounds which we term \emph{Bilinear Compressive Security (BCS)}.
The scheme is simple: In each round of communication, Alice randomly generates a vector $h\in \C^m$, called filter, that she convolves with the 
measurements $Qx$ to produce the effective cyphertext $y=h*Qx$. She then sends $y$ to Bob, who knows the key $Q$ but not the filter $h$. Bob thus needs to solve for \emph{both} filter $h$ and message $x$. This is the \emph{blind deconvolution problem}. Under some assumptions on $h$ and $x$, Bob will indeed be able to do so using efficient algorithms. 
We will discuss this 'correctness issue' in detail later in this paper. Importantly, the filter $h$ is a 'dummy-variable' that does not need to be communicated or agreed upon between Alice and Bob. 
In contrast to compressive security, our scheme needs no additional storage, and can be argued to be inherently supported in (mobile) communications scenarios. Indeed, when Alice sends a 'signal' $Qx$ to Bob, it will naturally go through a sparse convolution 'channel' $h$ due to fading and multi-path. In some high mobility scenarios, it might be not too unrealistic to assume that $h$ is drawn randomly at each transmission. Notably, this physical channel can be easily incorporated in our analysis.
 
How secure is the above BCS scheme? We imagine an eavesdropper Eve, that can observe pairs $(x_k,y_k)$ (possibly with chosen $x_k$). 
Given enough observations, will she be able to reveal the key $Q$? Since the filters are drawn randomly, this should be harder than the corresponding linear problem.
The purpose of this article is to make a rather fundamental mathematical analysis of this chosen plaintext attack. Most importantly, we will not assume that Eves applies any specific algorithm to recover $Q$. Our main result (Theorem \ref{th:sparse_security}) will state that under non-restrictive assumptions on the distribution of the filter $h$ -- it must be symmetric, $\theta h \sim h$ for any unit scalar $\theta \in \C$ -- $Q$ will be \emph{perfectly} secure from known plaintext attacks \emph{if the $x_k$ always are $1$-sparse}. Sparsity of the $x_k$ will also make it impossible to recover $Q$ without additional information unless many $x_k$ are observed: I all $x_k$ are $s$-sparse, the map from $Q$ to $h*Qx_k$ will not be injective unless the number of vectors is larger than $\max(n,(n/s)^2)$.


\subsection{Related work} 
\emph{Physical layer security}
The BCS scheme falls roughly into the category of physical layer security in communication theory, i.e. it could be naturally benchmarked with i) channel-based key
generation schemes (see, e.g., \cite{wunder2018secure} for half-duplex and \cite{Vogt2018_TCOM} for full duplex or \cite{Tang2021} for reconfigurable antennas as a reference)
and ii) wiretap coding (see, e.g., \cite{Tyagi2015_PIEEE,Wiese2017_TIT}. However, the former requires highly volatile channel conditions for sufficient key entropy which, as said above, is only true in high mobility scenarios while the latter needs a physical advantage of Bob and Alice relative to Eve which is also difficult to verify in practice \cite{Utkovski2019}. The only scheme that we are aware of which actually uses the blind deconvolution framework is our recent IEEE IFS paper \cite{wunder2025perfectly}. The focus there was however very different: It was assumed that $Q$ is public and $h$ is a physically reciprocal and essentially (unknown) channel. Secrecy comes then from the full duplex transmission and the considered so-called 'close talker' scenario which creates a physical superposition of almost equal strength signals at the attacker's device.

\emph{A note on related crypto problems}. The BCS setting has a conceptual resemblance to lightweight encryption schemes for IoT based on the \emph{Learning Parity with Noise} (LPN) problem, in particular the LPN-C cryptosystem~\cite{gilbert2008encrypt}. Both use linear structures with additive or multiplicative noise to obscure a low-complexity secret transformation. In both, the ciphertexts are linearly related to the plaintexts through a secret key (matrix) combined with random noise or filtering. The secrecy stems from the \emph{discrete algebraic hardness} of solving noisy linear equations over $\mathbb{F}_2$. The results we prove have the same underlying flavor: We show when it is impossible to invert the map from the key $Q$ to the 'noise-combined cyphertexts' $(h*Qx_k)_{k \in [M]}$. Since $Q$ here is an element of $\C^{m,n}$ (and not $\mathbb{F}_2^q$), the mathematical structure is however quite different.

\subsection{Structure of the article}
In Section \ref{sec:description}, we describe the encryption/decryption scheme together with the attack model.
Correctness and security analysis of the scheme are described in Section \ref{sec:correctness} and Section \ref{sec:security}, respectively. Numerical experiments are carried out in Section \ref{sec:experiment} confirming the empirical relevance of our security analysis. Then we conclude the paper in Section \ref{sec:conclusion}.


\begin{figure}[t]
    \centering
    \includegraphics[width=.45\textwidth]{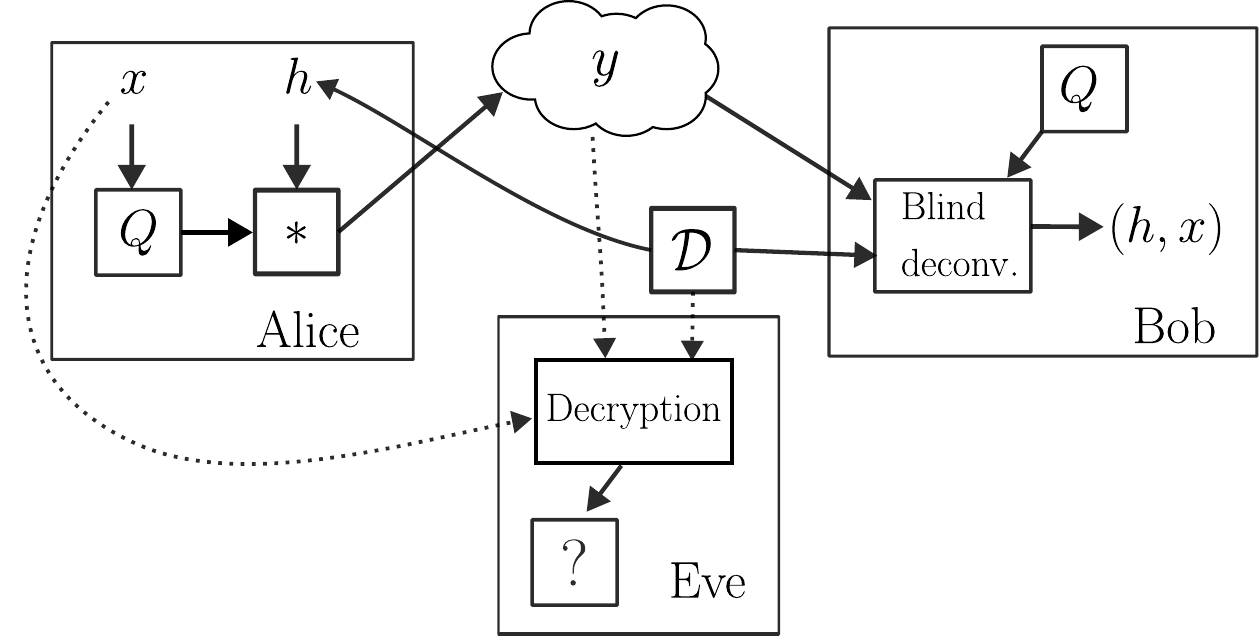}
    \caption{A graphical depiction of BCS. Note that $\calD$ and $y$ is public information, but $x$ and $Q$ are secret. 
    The cloud shall indicate a potential wireless channel which corroborates the security properties of BCS.}
    \label{fig:graphic_scheme}
\end{figure}

\section{System model} \label{sec:description}

Let us begin by formally describing the BCS scheme. It is illustrated in Figure 1. 

\subsection{Encryption and Decryption}

Alice wants to send a message $x\in \mathbb{C}^n$ to Bob. To do so, she first applies a matrix $Q\in \mathbb{C}^{m,n}$ where $Q$ is the secret key shared between Alice and Bob. Subsequently, she generates a random filter $h\in \mathbb{C}^{m}$ according to some (public) distribution $\calD$ and convolves it with $Qx$ to form the vector $y= h
*Qx$ representing the cyphertext which she then sends to Bob. Importantly, Alice does not communicate the concrete realization of $h$ to anyone, not even to Bob, and changes it for every transmission.

To decrypt $y$, Bob must solve the equation $y=h*Qx$ for both $h$ \emph{and} $x$, since both of them are unknown to him. This is the mentioned blind deconvolution problem. It is important to note that this problem has a fundamental scaling ambiguity: the pair $(h,x)$ will give the same value of $y$ as the pair $(h/\alpha,\alpha x)$ for any $\alpha\in \mathbb{C}$. Modulo this ambiguity, however, this problem is solveable under certain assumptions on $h$ and $x$  -- see Section \ref{sec:correctness}.



\subsection{Attack model}\label{subsec:attack-model}
We now formalize the adversarial scenario for BCS. Our goal is to derive universal lower bounds, particularly no a priori structure of $Q$ is asserted.

\begin{ass} 
Consider an adversary Eve that observes plaintext-cyphertext pairs $(x_k,y_k)$. We assume that 
\begin{itemize}
    \item Eve can observe completely noise-free cyphertexts.
    \item Eve can choose $M$ messages $x_k$, each of which she can additionally observe as many cyphertexts $y_{k\ell} = h_{k\ell} *Qx_k, \ell =0,1,2,\dots$ as needed.  That is, Eve gains access to the \emph{probability distribution} of $h*Qx_k$ induced by the draw of $h\sim \calD$ for a number of plaintexts $x_k$.
    \item Eve has complete knowledge of the filter distribution $\calD$.
\end{itemize}
\end{ass}
Note that these assumptions are hat are highly favorable to Eve (yet unrealistic in practice). Also note that Eve has neither access to the key $Q$, nor the \emph{concrete realization} $h_{k\ell}$ of  the random filter $h\sim \calD$. Her goal is to recover 
the secret key \( Q \).
These assumptions are sufficient to make strong claims about the security of the BCS scheme as we will discuss in Section \ref{sec:security}. We now move on to the correctness and security analysis.

\section{Correctness} \label{sec:correctness}

For the scheme to be correct, Bob needs to be able to solve the blind deconvolution problem. A large body of literature has shown that this is possible under structural assumptions on $h$ and $x$, and with high probability for a reasonable random draw of the key $Q$. A tailored result recently derived by a subset of the authors of this paper is as follows:
\begin{theo} \cite{flinth2022bisparse}
    Fix $s\leq n$, $\sigma \leq m$, and let $Q\in \C^{m,n}$ be a random matrix with i.i.d. Gaussian entries. There exists a constant $K>0$ so that if $m\geq K s\sigma \cdot\mathrm{polylog}(n,m)$, with very high probability over the draw of $Q$, Bob can use the efficient HiHTP algorithm \cite{hiHTP,wunder2018secure} to recover any pair of a $\sigma$-sparse $h$ and $s$-sparse $x$ from $y=h*Qx$.
\end{theo}

We refer to \cite{flinth2022bisparse} for detailed literature review of the bisparse blind deconvolution problem, including a description of both HiHTP as well as other recovery methods.  For us, the important message is that as long as Bob knows $Q$, and Alice sends messages and generates filters that are sparse, Bob will be able to recover them up to a scaling ambiguity. 

\section{Security analysis} \label{sec:security}
The rest of the paper will be concerned with the following question: \emph{Which properties of $\calD$ and $(x_k)_{k \in [M]}$ imply  (e.g., how small must $M$ be) that recovery of key $Q$ becomes impossible in BCS?}


\subsection{Preliminaries}
We begin by formally defining what we mean by recovery of the key. 
Again,  a fundamental scaling ambiguity needs to be take into account -- for a fixed filter, the message $x$ will lead to the same cyphertext for the key $Q$ as $\theta \cdot x$ will lead to for the key $Q/\theta$. While knowledge of the typical norm of the filter $h$ and $x$ should be able to resolve the magnitude of $\theta$, the phase will not. We therefore say that Eve has revealed $Q$ if she has done so modulo a \emph{unit scalar}.

\begin{defi} \label{def:recoverability}
    \begin{itemize}
        \item[(a)]For a complex vector space $U$, we denote by $\mathcal{P} U$ the space of equivalence classes in $U$ modulo unit scalars: $u\sim v$ if $u=\theta v$ for a $\theta\in \mathbb{C}$ with $\vert{\theta}\vert=1$.
        \item[(b)]  Let $h$ have a known distribution $\calD$. We say that a set of vectors $(x_k)_{k \in [M]}$ \emph{can retrieve the key for $\calD$-distributed $h$} 
        if the map $Q \to (h*Qx_k)_{k \in [M]}$ from $\mathcal{P} \C^{m,n}$ to the set of probability distributions on $(\C^{m})^M$ is injective.
    \end{itemize}   
\end{defi}


The recoverability of $Q$ depends on the distribution of the filter $h$. In this work, we consider two distinct sets of assumptions on $h$:
(a) $h$ is drawn from a standard Gaussian distribution, i.e., it exhibits no structural properties whatsoever;
(b) $h$ has a symmetric global phase.
Assumption (a) is deliberately restrictive and is included primarily for illustrative purposes. In contrast, assumption (b) is much weaker and allows for a lot of potential structure for Eve to exploit.

\subsection{A illustrative toy case: Gaussian filter.} As advertised, we first consider the case  that the filter $h$ is Gaussian, meaning that its components $h(i), i=1, \dots, m$ are i.i.d. Gaussian distributed. We reiterate that this case is presented only for illustration purposes -- these filters will have full support with probability $1$, so that Bob will have severe trouble solving the blind deconvolution problem, as a simple dimension counting argument shows: He then needs to use $m$ equations to reveal the $m$ parameters of $h$ in addition to the parameters in $x$, which will be impossible. 

The Gaussian assumption makes the recoverability question easy.
\begin{theo} \label{th:gaussian}
    For i.i.d. normal Gaussian filters, no set of vectors $(x_k)_{k \in [M]}$ can retrieve the key.
\end{theo}
\begin{proof}
Since the Fourier transform is unitary, the Fourier transforms of the filters, $\widehat{h}$, will also be i.i.d. Gaussian. The convolution theorem states that $\widehat{h*Qx} = \widehat{h} \odot\widehat{Qx}$, where $\odot$ denotes the pointwise product. Since the entries of a standard Gaussian vector are independent and rotationally symmetric, this implies that the distribution of $\widehat{h}\odot \widehat{Qx}$ will never reveal more information about $Qx$ than the values of the amplitude vectors $\vert\widehat{Qx}\vert = \vert F(Q)x\vert$, where $F(Q)$ denotes the matrix obtained via applying the Fourier transform column-wise on $Q$. 

For any set of vectors $x_k$, the map $\varphi: Q\mapsto \vert F(Q)x\vert$ is however clearly not injective modulu a unit scalar: for an arbitrary $W$, and $W'$ formed by multiplying the rows of $W$ with different unit scalars, we have $\varphi(F^{-1}W)=\varphi(F^{-1}W')$.
\end{proof}

The proof of Theorem \ref{th:gaussian} is illustrative. Its idea is that when the filter is Gaussian, the distributions does not reveal more than the \emph{phase-less} measurements $\abs{F(Q)x_k}\in \R^{m}$, $k\in [M]$ of $Q$. Let us decipher this slightly. Denoting $fq_i \in \C^n$, $i\in [m]$ as the \emph{rows of $F(Q)$}, we can write $\vert{F(Q)x_k}\vert$ as
\begin{align*}
    \abs{\sprod{fq_i,x_k}}, \quad i \in [m], k \in [M].
\end{align*}
Hence, to recover $Q$, Eve will need to solve $m$ independent \emph{phase retrieval problems}\cite{bandeira2014saving}.
\begin{defi}
    Let $(x_k)_{k \in [M]}$ be a set of vectors in $\C^n$. The problem of recovering a vector $q\in \C^n$ from the scalars $\abs{\sprod{q,x_k}}$ is called the \emph{phase retrieval problem}. We say that the set of vectors \emph{do phase retrieval} if the map $\phi: \C^n \to \R^M, q \mapsto (\abs{\sprod{q,x_k}})_{k \in [M]}$ is injective modulu unit scalars.
\end{defi}
One should stress that the main reason for the inability of the vectors $(x_k)$ to retrieve the key {is} \emph{of course not}  that phase retrieval is an unsolvable problem. Indeed, phase retrieval  has been very well studied, and there are multiple practical algorithms to solve it, such as PhaseLift and Wirtinger Flow.\cite{candes2015phase,candes2015wirtinger}. In particular, it has been shown that any generic choice of $M$ vectors can do phase retrieval in $\C^n$ when $M \geq 4n-4$ \cite{bandeira2014saving}.

As our proof of Theorem \ref{th:gaussian} reveals, the reason for non-recoverability of the key is instead the fact that the phase retrieval problems are \emph{independent} of each other. Because of this, each $q_i$ will be retrievable to a row-individual phase, but this does not determine $Q$ up to a global phase.

\subsection{A minimal filter assumption: Phase symmetry} Let us now consider a more realistic, and very non-restrictive, assumption on the filter distribution.
\begin{defi}
    We call a filter distribution \emph{phase-symmetric} if for each unit scalar $\theta \in \C$, $\theta h \sim h$.
\end{defi}
This assumption is for instance fulfilled when Alice draws $h$ by first determining a sparse support,
and then filling the non-zero entries with i.i.d. Gaussians. In fact, this assumption entails Alice using the identity filter $e_0$ multiplied with a random unit scalar, so that the $Qx$ are only minimally confounded. Still, as we will see, as long as Alice also makes sure to send sparse messages, Eve will still have a very hard time to recover $Q$.

The phase symmetry of $h$ immediately implies the following: If $Q$ and $Q'$ are two linear maps so that $Qx_k = \theta_k Q'x_k$ for some unit scalars $\theta_k$,  the variables $Y(Q,x_k)$ and $Y(Q',x_k)$ will be identically distributed. 
Eve will consequently not be able to recover $Q$. Let us record this as a lemma:

\begin{lem} \label{lem:injectivePhi}
    If $h$ is phase-symmetric, a set of vectors $(x_k)_{k \in [M]}$ can retrieve the key only if the map \\ $\Phi: \calP \C^{m,n} \to (\calP\C^m)^M, Q\mapsto (Qx_k)_{k\in [M]}$ is injective.
\end{lem}

\begin{rem} It should be noted that Lemma \ref{lem:injectivePhi} is a one-way statement: We do not claim that the injectivity of $\Phi$ is \emph{sufficient} to make the key retrievable. In fact, Theorem \ref{th:gaussian} gives a simple counterexample: the standard Gaussian $h$ is phase-symmetric, but the theorem applies irregardless of whether $\Phi$ is injective. We also do not claim that it always is possible to calculate $\Phi(Q)$ from samples of the $Y(Q,h_k)$.
\end{rem}

Knowing a vector $\nu \in \calP\C^m$, i.e. modulu a unit scalar, is equivalent to knowing the vector of its amplitudes, $\abs{\nu}$ and the \emph{relative} phases $\nu_i/\nu_j$. A version of this problem was studied in \cite{kornprobst2021phase}. 
A practical recovery algorithm was proposed and shown to be effective. However, in our language, no theoretical guarantees on the map $\Phi$ were derived for recovery to be possible. Maybe somewhat surprisingly, we have the following remarkably simple characterization.

\begin{theo} \label{th:CBRisPR}
    For a set of vectors $(x_k)_{k \in [M]}$, the map $\Phi$ from Lemma \ref{lem:injectivePhi} is injective if and only if $(x_k)_{k \in [M]}$ can do phase retrieval.
\end{theo}
The proof of Theorem \ref{th:CBRisPR} is quite mathematically involved. It in particular invokes non-trivial theory of both phase retrieval as well as so called \emph{generalized phase retrieval} from \cite{bandeira2014saving} and \cite{wang2019generalized}. Due to space limitations, we postpone it to full version. 

Theorem \ref{th:CBRisPR} together with some fundamental phase retrieval theory now lets us prove our main result.

\begin{theo} \label{th:sparse_security}
    Let $s\leq m$ and . If $h$ is phase symmetric, there are no sets of $s$-sparse vectors $(x_k)_{k \in [M]}$ that retrieve the key
    \begin{enumerate}[(a)]
    \item at all if $s=1$.
    
    \item with $M\leq \max(\tfrac{n(n-1)}{s(s-1)}, 4n-3-2\log_2(n-1))$ for $s\geq 2$.
    \end{enumerate}
\end{theo}
\begin{proof}
By Theorem \ref{th:CBRisPR}, it is enough to argue that the $(x_k)_{k\in [M]}$ cannot do phase retrieval. First, it is well known that sets with $M\leq 4n-3-2\log_2(n-1)$
 cannot do phase retrieval -- see e.g. \cite[Theorem~3]{heinosaari2013quantum}. 
 
 We thus only have to show that recovery is impossible for $s=1$ and for $M\leq \tfrac{n(n-1)}{s(s-1)}$ for $s\geq 2$. To do this, we by a fundamental result about phase retrieval (\cite[Lemma~9]{bandeira2014saving})) only need to argue that there exists a self-adjoint matrix $H \in \C^{m,m}$ of rank $1$ or $2$ with $\sprod{H,x_kx_k^*}=0, k \in [M]$.
 
    (a) If $x_k$ are all $1$-sparse, the matrices $x_kx_k^*$ are all diagonal matrices (with single entries). Hence, all zero-diagonal matrices $H$ have the property $\sprod{x_kx_k^*,H}=0$. Since there surely exist zero-diagonal self-adjoint $H$ with rank $2$ for $m\geq s \geq 2$, the claim is proven. 

    (b) Consider the set $P$ of pairs $(i,j)$ with $i\neq j$, for $k \in [M]$ the sets $P \supseteq E_k= \{(i,j) \, \vert \sprod{x_kx_k^*,e_ie_j^*}=0\}$, and $E= \cap_kE_k$.  Since the $x_kx_k^*$ are Hermitian, all $E_k$ and therefore also $E$ are symmetri: $(i,j)\in E \iff (j,i)\in E$. Consequently, if $E$ is non-empty, there will be a Hermitian matrix $H=e_ie_j^*+e_je_i^*$ of rank $2$ with $\sprod{x_kx_k^*,H}=0$ for $k\in [M]$. 

    $E$ being non-empty is equivalent to the complement $E^c$ not being all of $P$. Note that $E^c = \cup_{k \in [M]}E_k^c$.  Since the $x_k$ are sparse, the matrices $x_kx_k^*$ are also -- they all have at most $s(s-1)$ non-zero entries.  Consequently, if $M < \frac{n(n-1)}{s(s-1)}$.
, $\abs{E^c} \leq \sum_{k\in [M]}\abs{E_k^c} \leq M\cdot s(s-1) <n(n-1) = \abs{P}$. That is, $E^c\neq P \impl E \neq \emptyset$. The proof is finished.

\end{proof}

This result deserves a discussion. First,  \ref{th:sparse_security}(a), i.e. that the key $Q$ is not retrievable at all when $1$-sparse vectors are used -- might at first seem like a very strong result. It is indeed interesting from a mathematical point of view, but in terms of the security of the encryption scheme, it is however of limited importance. If Eve can calculate $\Phi(Q)$, she will in fact \emph{in that case} not need to retrieve the key, i.e determine $Q$ up to a global scalar, in order to determine the $x_k$. In fact, since $Qx_k$ in the case of $1$-sparse plaintexts always will be scaled single columns of $Q$, it is enough for her to determine each column up to an individual phase. This is easier than recovery of the key in the sense of Definition \ref{def:recoverability}, and is not even ruled out by Theorem \ref{th:gaussian} in the case of Gaussian filters! 

The much more interesting result is therefore \ref{th:sparse_security}(b). It quantifies the intuition that the key is more secure in our BCS scheme than the standard compressed sensing scheme. Where as $n$ vectors $x_k$ trivially are enough to retrieve the key in the latter case, we will need a quadratic amount if the $x_k$ are sparse! The security gain however plateaus quite early -- when $s \sim \sqrt{n}$, the results applicable to non-sparse $x_k$ dominate our guarantee.

\section{Experiments} \label{sec:experiment}

 Let us perform a numerical experiment to show that Theorem \ref{th:sparse_security} is of practical relevance.
   

\paragraph{Experiment description} Our results above are all injectivity statements, i.e. conditions ensuring that the equations we consider have a unique solution. They say nothing about whether a concrete algorithm can or can't be used to solve for the key for certain number of measurements. To test this, we perform a small numerical study. Let us stress that we here also only address the problem of retrieving a $Q\in \C^{m,n}$(up to a scalar) from the measurements $b=\Phi(Q)$, rather than estimating $Q$ directly from samples $h*Qx_k$. See also the discussion  of Theorem \ref{th:sparse_security}. 

Our approach to retrieve $Q$ from $\Phi(Q)$ is to fit the values $\Phi(Q)_k= Qx_k$ to the values $b_k$, taking the global phase ambiguity for each vector into account. That is, to minimize the loss function 
\[
\mathcal{L}(Q) = \sum_k \min_{\abs{\theta_k}=1}\left\| \Phi(Q)_k - \theta_k b_k \right\|_2^2.
\]
The value of $\mathcal{L}$ can be calculated although we only know the $b_k$ modulu global unit scalars. To minimize $\calL$, we use the L-BFGS-B optimization algorithm (as implemented in numpy) \cite{liu1989limited}, a quasi-Newton method designed for high-dimensional problems with variable constraints. Unlike full-memory quasi-Newton methods, which require storing and updating the entire Hessian matrix, L-BFGS-B uses a limited-memory approach, retaining only a small subset of gradient and variable updates, making it a good choice for solving large-scale optimization problems. Its effectiveness in addressing similar relative phase retrieval problems has been demonstrated in prior research, which justifies its selection for this study\cite{kornprobst2021phase,Ji2016phase}. 

The $x_k\in \C^n$ are generated i.i.d., by choosing supports of size $s$ uniformly at random, and then fill the supports with standard Gaussian entries. $Q\in \C^{m,n}$ is generated according to a Gaussian distribution, and then attempted to be recovered from $b \in (\calP \C^m)^M$ via minimization of $\calL$. For $m=3$ and $n=100$, as well as $m=5$ and $n=50$, we vary the number $M$ and the sparsity $s$ of $x_k$: For $n=50$, we test values for $M$ from $10$ to $500$ and $s$ between $1$ and $48$. For $n=100$ we test values of $M$ ranging from $50$ to $2000$ and $s$ between $1$ and $98$. In both cases, the values more tightly spaced for lower values
to  enable a test of a wide range of values for both $M$ and $s$ without costing too much computation time. Each experiment is repeated $100$ times, and we deem an attempt successful if the relative error modulo unit scalars. between $Q$ and $Q_{\text{recovered}}$, $\min_{\abs{\theta}=1} \|Q - \theta Q_{\text{recovered}}\|/{\|Q\|}$  is below $0.1$. We use this relatively lax threshold to save on computing time, but also considering that some type of quantization would probably be applicable in practice.

\begin{figure}[th]
    \centering
    \includegraphics[width=0.38\textwidth]{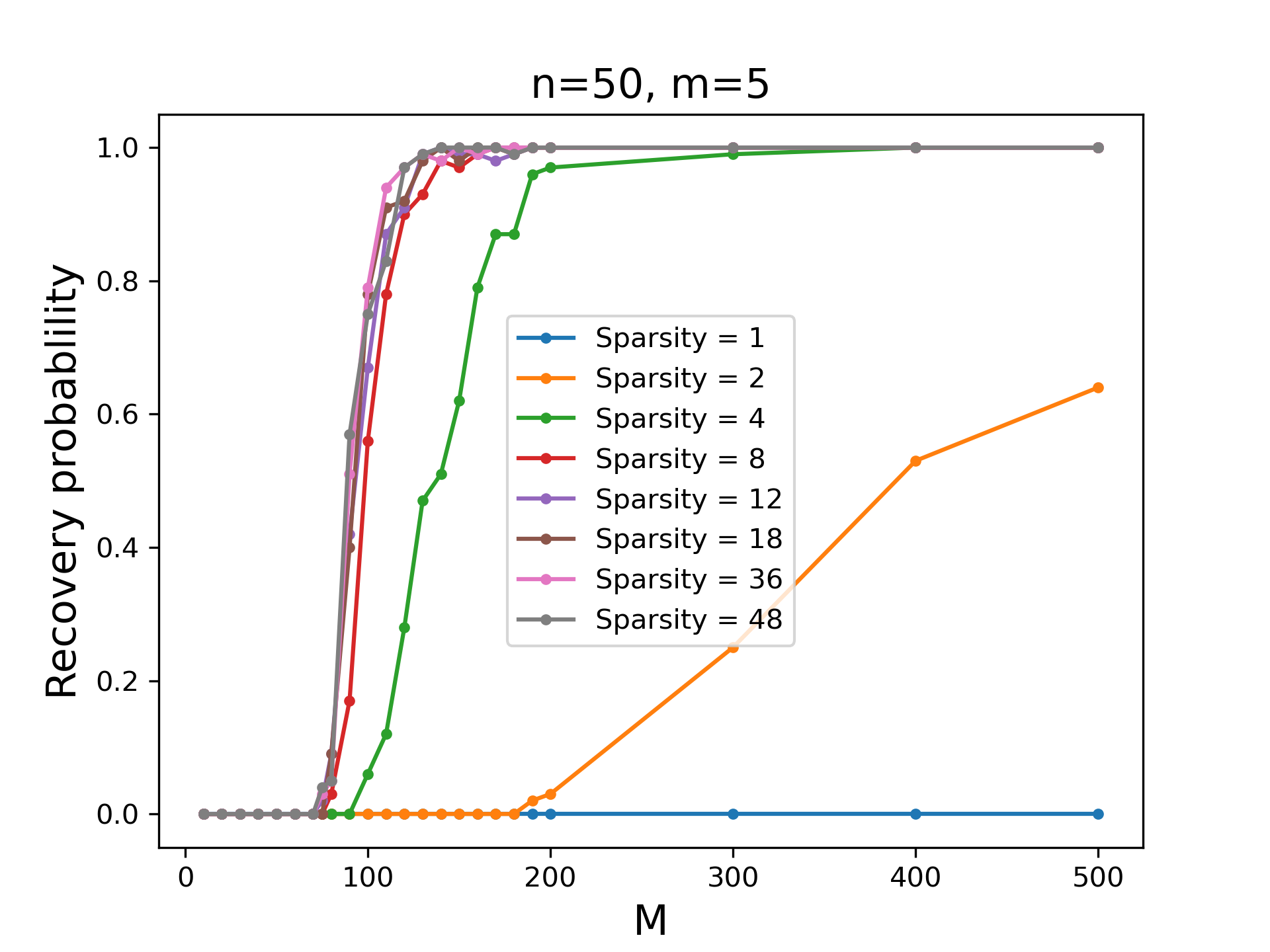} \includegraphics[width=0.38\textwidth]{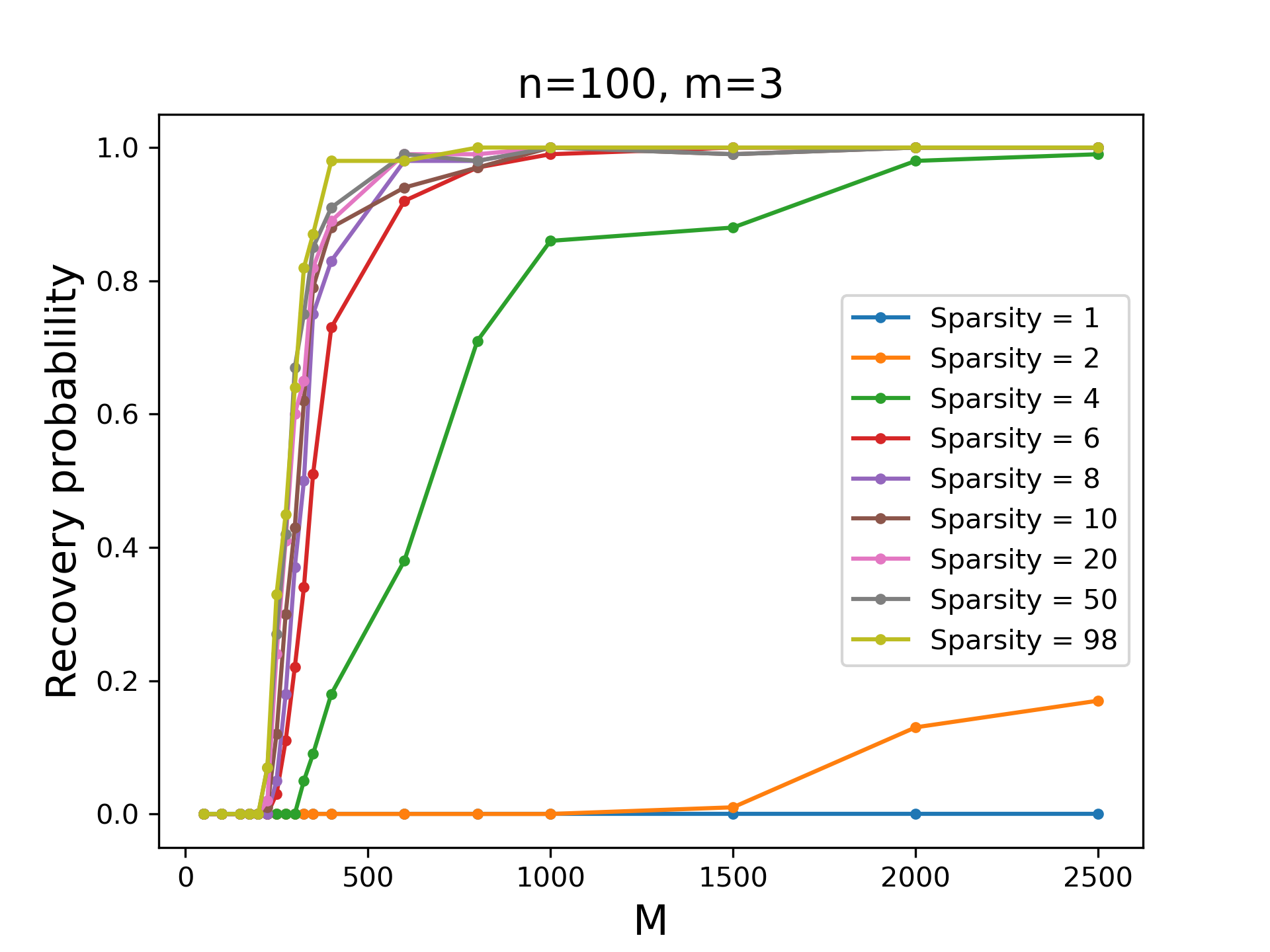}
    \caption{Success rate vs. $s$ and $M$ for $n=50$ and $m=5$ (top) and $n=100$ and $m=3$ (bottom). Best viewed in color.} 
    \label{fig:lineplots}
\end{figure}

\begin{figure}[th]
    \centering
    \includegraphics[width=0.4\textwidth]{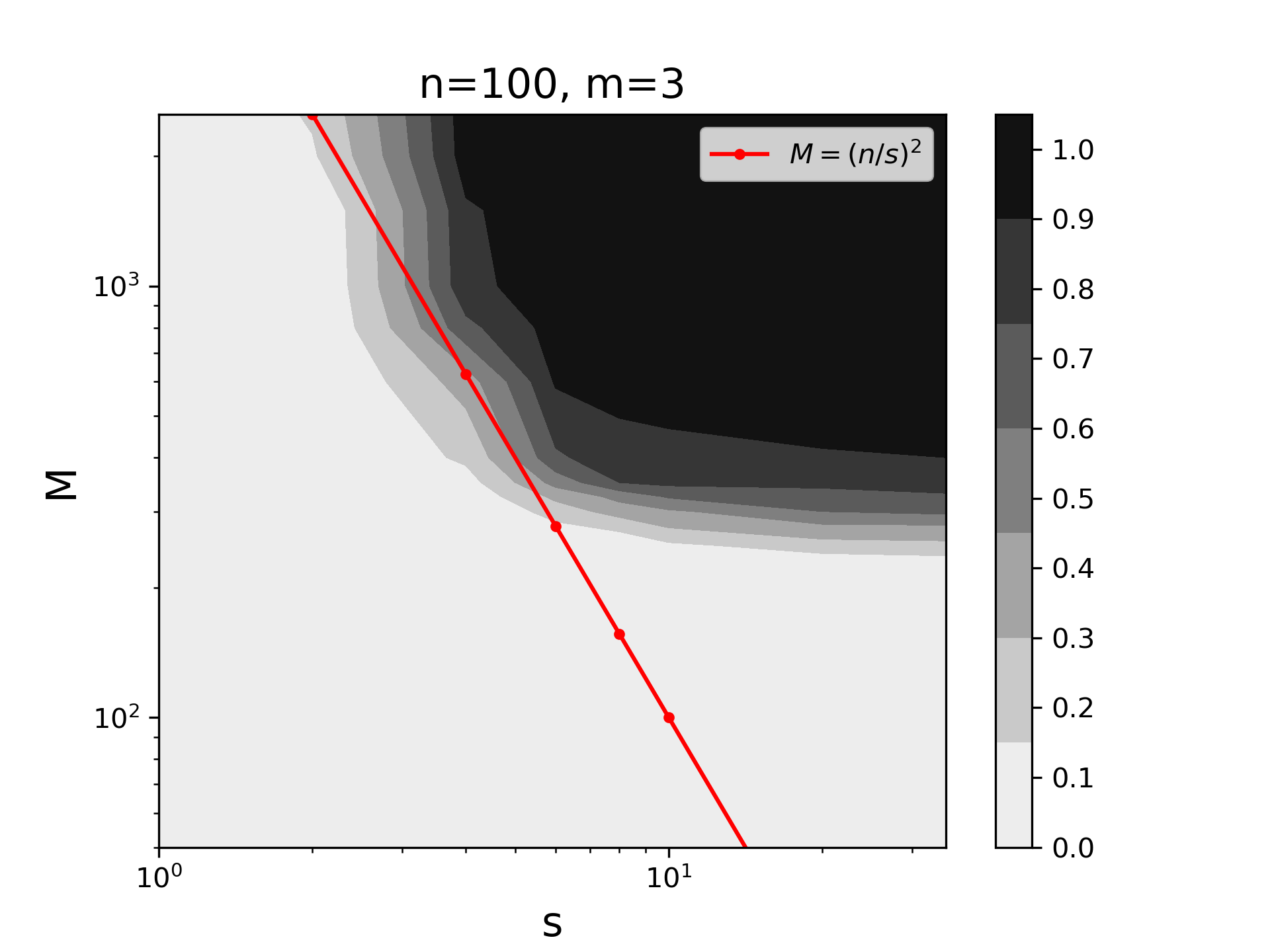}
    \caption{Success rates in dependence of $s$ and $M$. Both axes are in log-scale. } 
    \label{fig:contourplot}
\end{figure}

\paragraph{Results.}
Figure~\ref{fig:lineplots} shows the empirical recovery probability as a function of the sparsity $s$ and the number of measurements $M$. 
We observe that recovery fails for small values of $s$ and is in particular impossible for $s=1$. 
As $s$ increases, the recovery probability increases before plateauing around $s \approx 6$ for $n = 50$ and $s \approx 8$ for $n = 100$. 
These results are in correspondence with our theoretical predictions: recovery is impossible in the extreme sparse case of $s=1$ and becomes progressively easier as  $s$ increases, reaching a plateau roughly at $s \approx \sqrt{n}/2$.

To further strengthen the last point, in Figure \ref{fig:contourplot} we present a contourplot of the recovery probability in dependence of $M$ and $s$ for the $n=100$ experiments, along with the curve $L=(n/s)^2$. Here, we clearly see that the recovery threshold follows $(n/s)^2$ until a critical value for $s\approx 6 \approx \sqrt{n}/2$ , where the bounds applicable for non-sparse $x_k$ become relevant. This is in perfect agreement with the picture Theorem \ref{th:sparse_security} paints.


\section{Conclusion} \label{sec:conclusion} In this article, we proposed a new encryption scheme based on bilinear compressive sensing. The idea is simple: in addition to a linear encryption with a key $Q$, the sender convolves with a randomly drawn filter. Since the filter is redrawn at every transmission, the scheme circumvents the obvious known-plaintext-vulnerability of linear compressed sensing. We confirmed this intuition with a theoretical result: If Eve wants to recover $Q$ from studying the cyphertexts corresponding to $M$ $s$-sparse plaintexts $x\in \C^n$, $M$ needs to grow as $(n/s)^2$, even if she studies the distribution of the cyphertexts for an indefinite amount of draws of the filter. 

There is much room to prove stronger results. We have not given a practical method for Eve to recover $Q$ that provably works for realistic distributions for $h$. Also, we have not studied the effect of the fact that Eve in practice only can study finitely many samples of the cyphertext distribution. On the other hand, we have also not considered situations in which Eve has additional information about the key. Most importantly, we have only studied the problem of recovering $Q$ exactly, and not to which extent Eve can recover the plaintexts only from partial information of $Q$. These are interesting directions for future work.

\subsection*{Acknowledgement} AF acknowledges support
from the Wallenberg AI, Autonomous Systems and Software Program (WASP) funded by the Knut and
Alice Wallenberg Foundation. HO performed his research during an internship at the department of mathematics and mathematical statistics at Umeå University, and acknowledges support from them.
\bibliographystyle{abbrv} 
\bibliography{wdh}

\end{document}